\documentclass[a4paper,11pt]{article}
\usepackage{amssymb,amsmath,latexsym,amsthm}
\usepackage{color,mathrsfs}
\usepackage{url}
\usepackage{enumerate}
\usepackage[square,sort,comma,numbers]{natbib}
\usepackage{dsfont}
\usepackage{diagbox}

\usepackage{colortbl}

\usepackage{endnotes}

\usepackage[dvips]{graphics}
\usepackage[xetex]{graphicx}
\usepackage{fancyhdr}
\usepackage{pstricks,pst-plot,pst-node,pstricks-add}

\newtheorem{thm}{Theorem}[section]

\newtheorem{corollary}[thm]{Corollary}
\newtheorem{prop}[thm]{Proposition}
\newtheorem{Conjecture}[thm]{Conjecture}
\newtheorem{lemma}[thm]{Lemma}

\theoremstyle{definition}
\newtheorem{remark}[thm]{Remark}

\newtheorem{example}[thm]{Example}

\newcommand{\R}{\mathbb R}

\newcommand{\Z}{\mathbb Z}
\newcommand{\N}{\mathbb N}

\numberwithin{equation}{section}

\def\XXint#1#2#3{{\setbox0=\hbox{$#1{#2#3}{\int}$}
    \vcenter{\hbox{$#2#3$}}\kern-.5\wd0}}

\allowdisplaybreaks
\date{date}
\begin{document}
\title{Local optimality of cubic lattices for interaction energies}
\author{Laurent B\'{e}termin\thanks{\textit{E-mail address}: betermin@math.ku.dk}\\ \\ QMATH, Department of Mathematical Sciences,  \\ University of Copenhagen, \\ Universitetsparken 5,\\
DK-2100 Copenhagen \O,\\
Denmark }
\date\today
\maketitle

\begin{abstract}
We study the local optimality of Simple Cubic, Body-Centred-Cubic and Face-Centred-Cubic lattices among Bravais lattices of fixed density for some finite energy per point. Following the work of Ennola [Math. Proc. Cambridge, 60:855--875, 1964], we prove that these lattices are critical points of all the energies, we write the second derivatives in a simple way and we investigate the local optimality of these lattices for the theta function and the Lennard-Jones-type energies. In particular, we prove the local minimality of the FCC lattice (resp. BCC lattice) for large enough (resp. small enough) values of its scaling parameter and we also prove the fact that the simple cubic lattice is a saddle point of the energy. Furthermore, we prove the local minimality of the FCC and the BCC lattices at high density (with an optimal explicit bound) and its local maximality at low density in the Lennard-Jones-type potential case. We then show the local minimality of FCC and BCC lattices among all the Bravais lattices (without a density constraint). The largest possible open interval of density's values where the Simple Cubic lattice is a local minimizer is also computed.
\end{abstract}

\tableofcontents

\section{Introduction}

The mathematical justification of the emergence of crystal structures in the solid state of matter is something really difficult to get (see the recent review of Blanc and Lewin \cite{Blanc:2015yu}). Even in the zero temperature case, for enough heavy atoms where the interacting energy between the atoms can be restricted to a potential energy (see e.g. \cite{CazorlaBoronatQC} for the difference between classical and quantum solids), it is challenging to prove that the configuration with the lowest energy is a periodic lattice. More precisely, a type of crystallization problem is to prove, for a given function $f:(0,+\infty)\to \R$, that
\begin{equation}\label{thermolim}
\lim_{N\to +\infty}\frac{1}{N}\min_{X_N}\left\{\sum_{i\neq j \atop x_i,x_j\in X_N} f(|x_i-x_j|^2)  \right\}=\sum_{p\in L\backslash\{0\}} f(|p|^2),
\end{equation}
where $X_N=\{x_1,...,x_N\}\subset \R^d$,  for a certain lattice $L\subset \R^d$ that we can call an ``asymptotic minimizer in the sense of the thermodynamic limit". We do not know any result of type \eqref{thermolim} in dimension $d=3$ for classical potentials, but S\"uto proved interesting results at high density for oscillating potentials \cite{Suto2} and the optimality of the Face-Centred-Cubic (FCC, also written $D_3$) lattice in the sense of the thermodynamic limit, but with an additional three-body potential was proved by Theil and Flatley in \cite{TheilFlatley}. In dimension $d=2$, the asymptotic optimality of the triangular lattice $\Z(1,0)\oplus \Z (1/2,\sqrt{3}/2)$ was proved for some perturbations of the hard-sphere potential \cite{Rad2,Rad3,Crystal} and for oscillating potentials \cite{Suto1,Suto2}.\\

Another approach is to consider the minimization problem of
$$
L\mapsto E_f[L]:= \sum_{p\in L\backslash\{0\}} f(|p|^2),
$$
among periodic lattices $L\subset \R^d$. Hence, assuming periodicity and a fixed density of points, we look for the lattice with the lowest energy. This gives the best candidate among periodic structures for the associated crystallization problem \eqref{thermolim}. As we recalled in \cite{Beterloc}, this study was originally done by Rankin \cite{Rankin}, Cassels \cite{Cassels}, Ennola \cite{Eno2} and Diananda \cite{Diananda} in dimension $d=2$ for $f_s(r)=r^{-s}$, i.e. for the Epstein zeta function defined by
\begin{equation}\label{zetadef}
 L\mapsto \zeta_L(2s)=\sum_{p\in L\backslash\{0\}} \frac{1}{|p|^{2s}},\quad s>0,
 \end{equation}
and by Montgomery \cite{Mont} for $f_\alpha(r)=e^{-\alpha r}$, i.e. for the theta function defined by
\begin{equation}\label{thetadef}
L\mapsto \theta_L(\alpha)=\sum_{p\in L} e^{-\alpha |p|^2},\quad \alpha>0.
\end{equation}
They proved that the triangular lattice is the unique minimizer of these functions for any $s>0$ and any $\alpha>0$. In \cite{Betermin:2014fy,BetTheta15}, we solved this minimization problem, in dimension $d=2$, for several potentials $f$, and in particular for the Lennard-Jones-type potentials $f_{LJ}(r)=a_2 r^{-x_2}-a_1r^{-x_1}$, with energy
\begin{equation}
E_{f}[L]=a_2\zeta_L(2x_2)-a_1\zeta_L(2x_1),\quad (a_1,a_2)\in (0,+\infty),\quad x_1<x_2.
\end{equation}
More precisely, we proved the optimality and the non-optimality of the triangular lattice with respect to the density. The three-dimensional case was investigated by Ennola for the Epstein zeta function in \cite{Ennola}  where he proved the local optimality of the FCC lattices and the Body-Centred-Cubic (BCC, also written $D_3^*$) lattice, by duality, for $L\mapsto \zeta_L(2s)$, $s>0$, among Bravais lattices of fixed density. We studied this problem in \cite{BeterminPetrache} for the theta function, and we proved some local optimality and non-optimality results for the FCC and BCC lattices, with respect to the parameter $\alpha$. In the case of the Lennard-Jones-type potentials $f_{LJ}$, Born and Misra  \cite{born1940,misra1940} studied the local stability of the FCC, BCC and Simple Cubic lattices ($\Z^3$) with the lowest energy (among their dilated) among Bravais lattices. They proved the instability of $\Z^3$, the stability of the BCC lattice under some conditions on the exponents $(x_1,x_2)$ and the total stability of the FCC lattice.\\

This kind of result is interesting in a physical point of view. Indeed, among the 118 chemical elements of the periodic table, $21$ can have a BCC structure, $26$ can have a FCC structure and three can have a Simple Cubic structure\footnote{But only the Polonium can have this structure in its solid state at ambient temperature.}. Consequently, the knowledge of the stability, in the sense of local minimality, of the cubic structures for some lattice energies can help to understand the emergence of these structures, especially at zero temperature. Actually, the cooling of some materials give some final structures which can be a local minimizer (if the cooling is too fast for example) or a global minimizer of the potential energy. Obviously, these energies $E_f$ do not take into consideration any orbital of the atoms and have to be considered as a purely static toy models. However, BCC and FCC lattices play an important role in the theory of minimization of lattice energies among Bravais lattices, as shown by the following conjecture by Sarnak and Str\"ombergsson.

\begin{Conjecture}[\cite{SarStromb}]
Let $\zeta_L$ and $\theta_L$ be the Epstein zeta function and the theta function defined by \eqref{zetadef} and \eqref{thetadef}, then:
 \begin{enumerate}
 \item for any $\alpha> \pi$ and $s>3/2$,  $D_3$ is the unique minimizer of theta and Epstein zeta functions among Bravais lattices of fixed unit density.
 \item for any $\alpha< \pi$ and $s<3/2$, $D_3^*$ is the unique minimizer of theta and Epstein zeta functions  among Bravais lattices of fixed unit density.
 \end{enumerate}
\end{Conjecture}
Thus, in a certain sense, the BCC lattice seems to be a good candidate for energy minimization problems with long-range (in the sense of non-integrable at infinity) potentials, while this is the FCC if the potential is integrable at infinity. \\

As in \cite{Beterloc}, where we investigated the local optimality of the triangular and square lattices, we want to understand why the cubic lattices $\Z^3$, $D_3$ and $D_3^*$ are special in a stability point of view. We here focus on finite energies $E_f$ with finite first and second derivatives with respect to lattice parameters (we call the set of functions\footnote{See \eqref{defF} for more details.} $\mathcal{F}$). These $5$ parameters $(u,v,x,y,z)$ (because we fix the density) are chosen following Ennola's method \cite{Ennola}. Actually, several lemmas from this paper will be directly used or adapted to prove our results. In each statement, we will fix a volume $V$ of the primitive unit cell (the inverse of the density) and then $\Z^3$, $D_3$ and $D_3^*$ will correspond to Simple Cubic, FCC and BCC lattices with this volume $V$, to avoid extra non-useful notations\footnote{These notations $\Z^3$, $D_3$ and $D_3^*$ name the shape of the structure.}. The first result we get is the following fact

\begin{prop}[See Prop. \ref{criticality} below]
For any $f\in \mathcal{F}$ and any fixed volume $V$, $\Z^3$, $D_3$ and $D_3^*$ are critical points of $L\mapsto E_f[L]$ among Bravais lattices of fixed volume $V$.
\end{prop}
This means that $\nabla E_f[u,v,x,y,z]=0$ at points (lattices) $\Z^3$, $D_3$ and $D_3^*$. This result, proved by using the symmetries of each lattice (see Section \ref{secauto}), is obviously not surprising, and we obtained the same for the square and triangular lattices in dimension $2$ in \cite{Beterloc}.\\

Our second result follows from the study of the second order derivatives of $E_f$. As Ennola \cite{Ennola}, we use the symmetries of $D_3$ to write its Hessian matrix in terms of two ternary forms $R$ (associated to the FCC lattice) and $T$. These formulas are given in Proposition \ref{2ndD3} for any $f\in \mathcal{F}$. As Ennola showed that this decomposition is useful to prove the local optimality of $D_3$ for the Epstein zeta function, we do the same in the case of the theta function \eqref{thetadef}. Thus, using recent results by Faulhuber and Steinerberger \cite{Faulhuber:2016aa} about one dimensional theta functions, we get the following theorem.

\begin{thm}[See Prop. \ref{Z3saddleptLJ}, Prop. \ref{locmintheta}, Cor. \ref{corlocmintheta} below]\label{MainTh1}
Let $\theta_L$  be the theta function defined by \eqref{thetadef}. Then:
\begin{enumerate}
\item For any $\alpha>0$ and any $V>0$, $\Z^3$ is a saddle point of $L\mapsto \theta_L(\alpha)$ among Bravais lattices of fixed volume $V$. 
\item For any $V>0$, there exists $\alpha_0$ such that $D_3$ and $D_3^*$ are saddle points of $L\mapsto \theta_L(\alpha)$ respectively for $0<\alpha<\alpha_0$ and $\alpha>1/\alpha_0$ among Bravais lattices of fixed volume $V$.
\item For any $V>0$, there exists $\alpha_1$ such that $D_3$ and $D_3^*$ are local minimizers of $L\mapsto \theta_L(\alpha)$ respectively for $\alpha>\alpha_1$ and $0<\alpha<1/\alpha_1$ among Bravais lattices of fixed volume $V$.
\end{enumerate}
\end{thm}
\begin{remark}
We notice that, due to the scaling formula $\theta_{L}(\alpha)=\theta_{L_1}(V^{2/3}\alpha)$, for any Bravais lattice $L$ with volume $V$ such that $L=V^{1/3}L_1$ where $L_1$ has unit volume, it will be sufficient to prove this result for $V=1$. Furthermore, $\alpha_0$ and $\alpha_1$ depend on $V$.
\end{remark}

In particular, we improve our previous result  \cite[Theorem 1.7.3)]{BeterminPetrache} where we only proved the non-optimality of these lattices for the extremal values of $\alpha$. In \cite{BeterminPetrache}, the local optimality was just proved for a finite number of values for $\alpha$ thanks to a result about the global minimality of these lattices among body-centred-orthorhombic lattices. Our Theorem \ref{MainTh1} supports the conjecture of Sarnak and Str\"ombergsson. Furthermore, the first point of this theorem implies that, for any fixed volume $V$ and any completely monotone function\footnote{A function $f$ is completely monotone if, for any $k\in \N$ and any $r>0$, $f^{(k)}(r)\geq 0$ or, equivalently, if it is the Laplace transform of a positive Borel measure on $\R_+$.} $f\in \mathcal{F}$, $\Z^3$ is a saddle point of $E_f$ among Bravais lattices of fixed volume $V$.\\

As in \cite{Beterloc}, our final results are about the Lennard-Jones-type potentials. First, we find the largest open set of values of $V$ such that the cubic lattices are local minimizers/maximizers or saddle points. Second, we show the local minimality of FCC and BCC lattices at high density.

\begin{thm}[See Prop. \ref{Z3LJ} and Prop. \ref{locoptFCCLJ} below]\label{MainTh2}
Let $f$ be the Lennard-Jones-type potential defined on $(0,+\infty)$ by
\begin{equation}\label{defLJones}
f(r)=\frac{a_2}{r^{x_2}}-\frac{a_1}{r^{x_1}},
\end{equation}
where $(a_1,a_2)\in (0,+\infty)^2$ and $3/2<x_1<x_2$. Then:
\begin{enumerate}
\item If $a_1=2$, $a_2=1$, $x_1=3$, $x_2=6$, there exist $V_1,V_2$ such that, for any $V_1<V<V_2$, $\Z^3$ is a local minimizer of $E_f$ among Bravais lattices of fixed volume $V$, where $V_1\approx 1.200$ and $V_2\approx 1.344$. Furthermore, if $V\not\in [V_1,V_2]$, then $\Z^3$ is a saddle point of $E_f$  among Bravais lattices of fixed volume $V$.
\item Let $G$ and $H$ be two functions defined respectively by \eqref{G} and \eqref{H}. If we have
\begin{equation}
V<\frac{1}{\sqrt{2}} \min\left\{ \left( \frac{a_2G(x_2)}{a_1G(x_1)} \right)^{\frac{3}{2(x_2-x_1)}}, \left( \frac{a_2H(x_2)}{a_1H(x_1)} \right)^{\frac{3}{2(x_2-x_1)}} \right\},
\end{equation}
then $D_3$ and $D_3^*$ are local minimizers of $E_f$ among Bravais lattices of fixed volume $V$. Furthermore, if 
\begin{equation}
V> \frac{1}{\sqrt{2}} \max\left\{ \left( \frac{a_2G(x_2)}{a_1G(x_1)} \right)^{\frac{3}{2(x_2-x_1)}}, \left( \frac{a_2H(x_2)}{a_1H(x_1)} \right)^{\frac{3}{2(x_2-x_1)}} \right\},
\end{equation}
then $D_3$ and $D_3^*$ are local maximizers of $E_f$ among Bravais lattices of fixed volume $V$.
\end{enumerate}
\end{thm}
In particular, for Lennard-Jones-type interactions, both FCC and BCC lattices are local minimizers at high density and local maximizers at low density. For the classical Lennard-Jones case, i.e. $a_1=2$, $a_2=1$, $x_1=3$, $x_2=6$, we find a interval of values of $V$ such that $D_3$ and $D_3^*$ are saddle points of $E_f$. We notice that it is not the case in dimension $d=2$, where the triangular lattice is a local minimizer or a local maximizer for almost every values of the area of the lattices (see \cite{Beterloc}). Furthermore it turns out -- and this is not surprising -- that $\Z^2$ and $\Z^3$ have the same behaviour in terms of local minimality: they are saddle points for both theta and Epstein zeta functions and there exists an open interval where they are local minimizers for the Lennard-Jones-type energy. Moreover, using the analytic continuation of $\zeta_L$, this theorem stays true for $0<x_1<x_2$. From the previous result, we get the local optimality of the FCC and BCC lattices among all the Bravais lattices (without volume restriction):

\begin{prop}[See Prop. \ref{prop-locglobalLJ} below]\label{prop-MainLocLJGlob}
Let $f$ be defined by \eqref{defLJones} with $a_1=2$, $a_2=1$, $x_1=3$ and $x_2=6$. Let $V^*(\Z^3)^{1/3}\Z^3$ (resp. $V^*(D_3)^{1/3}D_3$ and $V^*(D_3^*)^{1/3}D_3^*$) be the minimizer of $E_f$ among all the dilated of $\Z^3$ (resp. $D_3$ and $D_3^*$) with volume $V^*(\Z^3)$ (resp. $V^*(D_3)$ and $V^*(D_3^*)$). Then we have:
\begin{enumerate}
\item  $V^*(\Z^3)^{1/3}\Z^3$ is not a local minimizer of $E_f$ among all the Bravais lattices;
\item $V^*(D_3)^{1/3}D_3$ and $V^*(D_3^*)^{1/3}D_3^*$ are local minimizers of $E_f$ among all the Bravais lattices.
\end{enumerate}
\end{prop}

Thus, for the global minimality of $E_f$, it is natural, from Proposition \ref{prop-MainLocLJGlob}, from the Sarnak-Str\"ombergsson conjecture, from our two-dimensional investigations \cite{Betermin:2014fy,BetTheta15,Beterloc} and from the stability results of Born et al. \cite{born1940,misra1940}, to state the following conjecture about the global minimality, i.e. without a volume constraint, of $D_3$ and $D_3^*$

\begin{Conjecture}
Let $f$ defined by \eqref{defLJones}, with $(a_1,a_2)\in (0,+\infty)^2$, then:
\begin{enumerate}
\item If $3/2<x_1<x_2$, then the minimizer of $E_f$ among all Bravais lattices is unique and it is a FCC lattice.
\item If $0<x_1<x_2<3/2$, then the minimizer of $E_f$ among all Bravais lattices is unique and it is a BCC lattice.
\end{enumerate}
\end{Conjecture}

Obviously, it is reasonable to think that a proof of Sarnak-Str\"ombergsson conjecture should be done before to be able to prove this one.\\

In Section \ref{secenergylattice}, we start with some preliminaries: we recall the Ennola's parametrization of the ternary quadratic forms, the definitions of lattices and energies and some important properties of the automorphs of the forms associated to Simple Cubic and Face-Centred-Cubic lattices. In Section \ref{seccritic} we prove the fact that the cubic lattices are critical points for any potential. In Section \ref{sec2ndorder} we compute the second order derivatives at lattices $\Z^3$ and $D_3$ thanks to general formulas of these derivatives given in Section \ref{secAppendix} and Lemma \ref{automorphs}. Thus, Theorems \ref{MainTh1} and \ref{MainTh2} are proved in Sections \ref{sectheta} and \ref{secLJ}, as well as Proposition \ref{prop-MainLocLJGlob}.

\section{Energies and lattices}\label{secenergylattice}

In this part, we show how to parametrize a Bravais lattice $L\subset \R^3$ from its associated quadratic form, following \cite{Ennola}. Thus, given an integrable potential $f$, we define its finite energy $E_f[L]$. Furthermore, the automorphs of the quadratic forms associated with $\Z^3$ and $D_3$ are used to simplify some lattice sums, as in \cite{Ennola}.

\subsection{Parametrization of a Bravais lattice and energy}
We parametrize any Bravais lattice $L=\Z v_1 \oplus \Z v_2 \oplus \Z v_3$ of fixed volume (for its primitive cell) $V$ by $(u,v,x,y,z)$ such that its associated quadratic form is given by
$$
Q_L(m,n,p)=\frac{V^{2/3}2^{1/3}}{u}\left[ (m+xn+yp)^2+v^2(n+zp)^2+\frac{u^3}{2v^2}p^2 \right].
$$
Throughout this paper, we will use the notation $C:=V^{2/3}2^{1/3}$. We notice that this remains to parametrize $v_1,v_2,v_3$ by
\begin{equation}\label{defv}
v_1=\sqrt{C}\left(\frac{1}{\sqrt{u}},0,0  \right),\quad v_2=\sqrt{C}\left(\frac{x}{\sqrt{u}},\frac{v}{\sqrt{u}},0  \right),\quad v_3=\sqrt{C}\left(\frac{y}{\sqrt{u}},\frac{vz}{\sqrt{u}}, \frac{u}{v\sqrt{2}}  \right).
\end{equation}
We define the set of potential $\mathcal{F}$ by
\begin{equation}\label{defF}
\mathcal{F}:=\left\{ f\in C^2((0,+\infty)); \forall k\in\{0,1,2\}, |f^{(k)}(r)|=O(r^{-3/2-k-\eta_k}), \textnormal{ for some }\eta_k>0 \right\}.
\end{equation}
Thus, for any $f\in\mathcal{F}$, we define the $f$-energy of $L$ by
$$
E_f[L]:=\sum_{p\in L\backslash\{0\}} f(|p|^2)=\sum_{m,n,p}f\left(\frac{C}{u}\left[ (m+xn+yp)^2+v^2(n+zp)^2+\frac{u^3}{2v^2}p^2 \right]  \right),
$$
where the sum is taken over all $(m,n,p)\in \Z^3 \backslash\{(0,0,0)\}$. Thus, let $L$ be a Bravais lattice of volume $V$ parametrize by $(u,v,x,y,z)$, then $L\mapsto E_f[L]$ can be viewed as a function of $5$ variables (the volume $V$ is fixed), i.e.
$$
E_f(u,v,x,y,z,V):=E_f[L].
$$
We will say that $L_0$, parametrized by $(u_0,v_0,x_0,y_0,z_0)$, is a critical point (lattice) of $E_f$ among the lattices of fixed volume $V_0$ if $(u_0,v_0,x_0,y_0,z_0,V_0)$ is a critical point of $E_f$. The same will hold for the local optimality of $L_0$ among Bravais lattices of fixed volume $V_0$.

Furthermore, we consider three important lattices, called ``cubic lattices":
\begin{enumerate}
\item the FCC lattice parametrized by $D_3=(u,v,x,y,z)=(1,1,0,1/2,1/2)$;
\item the BCC lattice, which is the dual $D_3^*$ of the FCC lattice;
\item the cubic lattice parametrized by $\Z^3=(u,v,x,y,z)=(2^{1/3},1,0,0,0)$.
\end{enumerate}

\begin{figure}[!h]
\centering
\includegraphics[width=9cm]{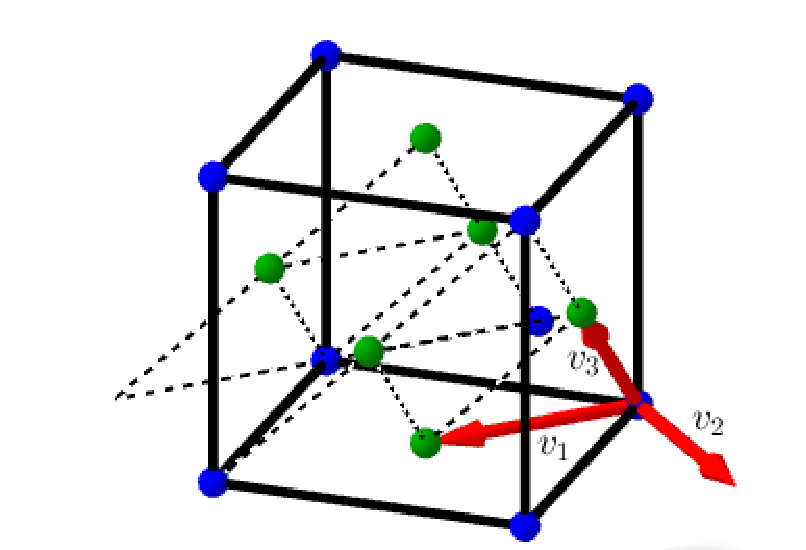}
\caption{Face-Centred-Cubic lattice $D_3=(1,1,0,1/2,1/2)$ generated by the basis (in red) $v_1=(\sqrt{C},0,0)$, $v_2=(0,\sqrt{C},0)$ and $v_3=(\sqrt{C}/2,\sqrt{C}/2,\sqrt{2C}/2)$. The green points are the center of the faces with blue vertices. The dash lines show how the FCC lattice is constructed layer-by-layer from the translated copies of a triangular lattice.}
\label{Parametrization}
\end{figure}
With the notation $I(m,n,p)=2^{-1/3}(m^2+n^2+p^2)$, we get
$$
Q_{V^{1/3}\Z^3}(m,n,p)=C I(m,n,p).
$$
We will write
$$
R(m,n,p)=m^2+n^2+p^2+mp+np,
$$
such that $Q_{V^{1/3}D_3}(m,n,p)=C R(m,n,p)$, and
$$
T(m,n,p):=mn(m+p)(n+p).
$$
We recall \cite[Eq. (18)]{Ennola}:
\begin{equation}\label{ineqRT}
4|T(m,n,p)|\leq R(m,n,p)^2.
\end{equation}
We will write $Q_L$, $I$, $R$ and $T$ instead of $Q_L(m,n,p)$, $I(m,n,p)$, $R(m,n,p)$ and $T(m,n,p)$.\\
Furthermore, we remark that
\begin{equation}\label{zetaR}
\zeta_{2^{-1/3}D_3}(2s)=\sum_{m,n,p}\frac{1}{R(m,n,p)^s}.
\end{equation}
We also define the following function
\begin{equation}\label{defY}
Y(s)=\sum_{m,n,p} \frac{T(m,n,p)}{R(m,n,p)^{s+2}}.
\end{equation}
We also recall the functional equation of the Epstein zeta function: for any $s>3/2$,
\begin{equation}\label{fcteqzeta}
\pi^{-s}\Gamma(s)\zeta_{L}(2s)=\pi^{-(3/2-s)}\Gamma(3/2-s)\zeta_{L^*}(3-2s).
\end{equation}

\subsection{Automorphs of the quadratic forms}\label{secauto}

First, it is totally clear that, for any $(m,n,p)\in\Z^3$, we have, for every permutation $\sigma$ of integers $(m,n,p)$,
\begin{equation}\label{scubicsym}
I(m,n,p)=I(\sigma(m,n,p))=I(-m,n,p)=I(m,-n,p)=I(m,n,-p).
\end{equation}

Second,  Ennola \cite{Ennola} showed that it is possible to use the automorphs of $R$, i.e. the matrices $A$ such that, for  any $(m,n,p)\in\Z^3$, 
$$
R(A(m,n,p))=R(m,n,p),
$$
in order to simplify the lattice sums of type $\sum_{m,n,p}Q(m,n,p)F(R(m,n,p))$ for any given function $F$ such that this sum converges. In the following lemma, we give the results of \cite[p. 861 and 864]{Ennola} in terms of lattice sums.

\begin{lemma}[\cite{Ennola}]\label{automorphs}
For any $F:\R_+\to \R$ such that the following sums are convergent, we have
\begin{align}
&\label{aut1} \sum_{m,n,p}(4n^2+4np-p^2)F(R)=0,\\
&\label{aut2} \sum_{m,n,p}n(2m+p)F(R)=0,\\
&\label{aut3}\sum_{m,n,p}p(2m+p)F(R)=0,\\
&\label{aut4}\sum_{m,n,p}n(2n+p)F(R)=0,\\
& \label{aut5} \sum_{m,n,p}p^2F(R)=\frac{2}{3}\sum_{m,n,p} R F(R),\\
&\label{aut6}\sum_{m,n,p}(4n^2+4np-p^2)^2 F(R)=\sum_{m,n,p}\left(\frac{10}{3}R^2+\frac{56}{3}T  \right)F(R),\\
&\label{aut7} \sum_{m,n,p}n^2 F(R)=\frac{1}{2}\sum_{m,n,p} R F(R),\\
&\label{aut8} \sum_{m,n,p}n^2(2m+p)^2F(R)=\sum_{m,n,p}\left( \frac{1}{3}R^2+\frac{4}{3}T \right)F(R),\\
&\label{aut9} \sum_{m,n,p}p^2(2m+p)^2F(R)=\sum_{m,n,p}\left( \frac{2}{3}R^2-\frac{8}{3}T \right)F(R),\\
&\label{aut10} \sum_{m,n,p}p^2(2n+p)^2F(R)=\sum_{m,n,p}\left( \frac{2}{3}R^2-\frac{8}{3}T \right)F(R),\\
&\label{aut11} \sum_{m,n,p} n(2m+p)(4n^2+4np-p^2)F(R)=0,\\
&\label{aut12} \sum_{m,n,p} p(2m+p)(4n^2+4np-p^2)F(R)=0,\\
&\label{aut13} \sum_{m,n,p} p(2n+p)(4n^2+4np-p^2)F(R)=0,\\
&\label{aut14} \sum_{m,n,p} npF(R)=-\frac{1}{3}\sum_{m,n,p} RF(R),\\
&\label{aut15} \sum_{m,n,p} np(2m+p)^2 F(R)=\sum_{m,n,p}\left( -\frac{R^2}{3}+\frac{4T}{3}  \right)F(R),\\
&\label{aut16} \sum_{m,n,p} np(2m+p)(2n+p)F(R)=0,\\
&\label{aut17} \sum_{m,n,p} p^2(2m+p)(2n+p)F(R)=0,\\
&\label{aut18} \sum_{m,n,p} p^4 F(R)=\sum_{m,n,p}\left(\frac{2}{3}R^2+\frac{8}{3}T  \right)F(R),\\
&\label{aut19}  \sum_{m,n,p}p^2(4n^2+4np-p^2)F(R)=-\sum_{m,n,p}\left( \frac{2}{3}R^2+8T \right)F(R).
\end{align}
\end{lemma}
Furthermore, it is clear that exchanging $m$ and $n$ lets $R$ invariant, i.e. for any $(m,n,p)\in \Z^3$,
\begin{equation}\label{exmn}
R(n,m,p)=R(m,n,p).
\end{equation}

\section{Criticality}\label{seccritic}

In this section, we prove the fact that $D_3$, $D_3^*$ and $\Z^3$ are critical points of $E_f$ for any $f$ and for any fixed volume $V$.

\begin{lemma}[First derivatives of the energy] For any $f\in \mathcal{F}$, any fixed volume $V>0$ and any $L=(u,v,x,y,z)$,
\begin{align*}
&\partial_u E_f[L]=C\sum_{m,n,p}\left( -\frac{(m+xn+yp)^2}{u^2}-\frac{v^2(n+zp)^2}{u^2}+\frac{u}{v^2}p^2  \right)f'\left( Q_L \right),\\
&\partial_v E_f[L]=C\sum_{m,n,p}\left(\frac{2v(n+zp)^2}{u}-\frac{u^2}{v^3}p^2  \right)f'\left( Q_L \right),\\
&\partial_x E_f[L]=\frac{2C}{u}\sum_{m,n,p} n(m+xn+yp) f'\left( Q_L \right),\\
&\partial_y E_f[L]=\frac{2C}{u}\sum_{m,n,p} p(m+xn+yp) f'\left( Q_L \right),\\
&\partial_z E_f[L]=\frac{2Cv^2}{u}\sum_{m,n,p} p(n+pz) f'\left( Q_L \right).
\end{align*}
\end{lemma}

\begin{prop}\label{criticality}
 For any $f\in\mathcal{F}$ and any fixed $V>0$, $D_3$, $D_3^*$ and $\Z^3$ are critical points of $L\mapsto E_f[L]$ among Bravais lattices of fixed volume $V$.
\end{prop}
\begin{proof}
Let $V>0$ be fixed. If $L=D_3=(1,1,0,1/2,1/2)$, then
\begin{align*}
&\partial_u E_f[D_3]=C\sum_{m,n,p}\left( -(m+p/2)^2-(n+p/2)^2+p^2 \right)f'\left(C R \right),\\
&\partial_v E_f[D_3]=C\sum_{m,n,p}\left(2(n+p/2)^2-p^2  \right)f'\left(C R  \right),\\
&\partial_x E_f[D_3]=C\sum_{m,n,p}n(2m+p)f'\left(C  R \right),\\
&\partial_y E_f[D_3]=C\sum_{m,n,p}p(2m+p)f'\left(C R \right),\\
&\partial_z E_f[D_3]=C\sum_{m,n,p}p(2n+p)f'\left(C R  \right).
\end{align*}
Thus, using \eqref{exmn}, we get
$$
\partial_u E_f[D_3]=-\frac{C}{2}\sum_{m,n,p}(4n^2+4np-p^2)f'(CR)=-\partial_v  E_f[D_3].
$$
Therefore, by \eqref{aut1}, we obtain $\partial_u E_f[D_3]=\partial_v  E_f[D_3]=0$. For the three other derivatives, we use respectively \eqref{aut2}, \eqref{aut3} and \eqref{aut4} and we get
$$
\partial_x E_f[D_3]=\partial_y E_f[D_3]=\partial_z E_f[D_3]=0.
$$
Thus, $D_3$ is a critical point of $E_f$ for any $f\in \mathcal{F}$. This implies the same for $D_3^*$, by the Poisson summation formula
$$
\sum_{p\in L} F(|p|)=|L|^{-1}\sum_{p\in L^*}\hat{F}(|p|),
$$
and the fact that the Fourier transform is an isomorphism on $\mathcal{F}$.\\
For the cubic lattice $\Z^3=(2^{1/3},1,0,0,0)$, we get, for any $V>0$ and any $f\in \mathcal{F}$,
\begin{align*}
&\partial_u E_f[\Z^3]=C\sum_{m,n,p}(-2^{-2/3}m^2-2^{-2/3}n^2+2^{1/3}p^2)f'(C I),\\
&\partial_v E_f[\Z^3]=2^{2/3}C\sum_{m,n,p} (n^2-p^2)f'(C I),\\
&\partial_x E_f[\Z^3]=2^{2/3}C\sum_{m,n,p}mn f'(C I),\\
&\partial_y E_f[\Z^3]=2^{2/3}C\sum_{m,n,p}mp f'(C I),\\
&\partial_z E_f[\Z^3]=2^{2/3}C\sum_{m,n,p}pn f'(C I).
\end{align*}
Using \eqref{scubicsym}, it is clear that all the partial derivatives are equal to $0$.
\end{proof}

\section{Second order derivatives}\label{sec2ndorder}

In the Appendix, we give the general formulas of all the second order derivatives of $E_f$ for any lattice $L=(u,v,x,y,z)$ with fixed volume $V$. In this part, we compute the second order derivatives of $E_f$ for $\Z^3$ and $D_3$.

\begin{prop}\label{2ndZ3}
For any fixed $V>0$ and any $f\in \mathcal{F}$, the second derivatives of $E_f$ at point $\Z^3$ are
\begin{align*}
&\partial_{uu}^2 E_f[\Z^3]=\frac{3}{2^{1/3}}C^2\sum_{m,n,p}(p^4-p^2n^2) f''(CI)+ 4C\sum_{m,n,p}p^2 f'(CI),\\
&\partial_{vv}^2 E_f[\Z^3]=2^{7/3}C^2\sum_{m,n,p}(p^4-n^2p^2) f''(CI) + 2^{8/3}C\sum_{m,n,p}p^2 f'(CI),\\
&\partial_{xx}^2 E_f[\Z^3]=\partial_{yy}^2 E_f[\Z^3]=\partial_{zz}^2 E_f[\Z^3]=2^{4/3}C^2\sum_{m,n,p}n^2p^2 f''(CI)+2^{2/3}C\sum_{m,n,p}p^2f'(CI),\\
&\partial_{uv}^2 E_f[\Z^3]=-2C^2\sum_{m,n,p}(p^4-n^2p^2)f''(CI)-3 \times 2^{1/3}C\sum_{m,n,p}p^2 f'(CI),\\
&\partial_{ux}^2 E_f[\Z^3]=\partial_{uy}^2 E_f[\Z^3]=\partial_{uz}^2 E_f[\Z^3]=\partial_{vx}^2 E_f[\Z^3]=\partial_{vy}^2 E_f[\Z^3]=\partial_{vz}^2 E_f[\Z^3]=\partial_{xy}^2 E_f[\Z^3]=\partial_{xz}^2 E_f[\Z^3]\\
&\quad\quad\quad\quad \hspace{1mm}=\partial_{yz}^2 E_f[\Z^3]=0.
\end{align*}
\end{prop}
\begin{proof}
The computation is straightforward by using Lemma \ref{2nd} and \eqref{scubicsym}.
\end{proof}

We now use Lemma \ref{automorphs} in order to get the following expression of the second derivatives of the energy in terms of $R$ and $T$. This decomposition, already did by Ennola \cite{Ennola} in the case of the Epstein zeta function, will be very useful to study the optimality of $D_3$ for the theta function. We do not give the details of our computations here.

\begin{prop}\label{2ndD3}
For any fixed $V>0$ and any $f\in \mathcal{F}$, the second derivatives of $E_f$ at point $D_3$ are
\begin{align*}
&\partial_{uu}^2 E_f[D_3]=\frac{C^2}{2}\sum_{m,n,p} R^2 f''(CR)+2C\sum_{m,n,p} Rf'(CR)+6C^2\sum_{m,n,p} T f''(CR),\\
&\partial_{vv}^2 E_f[D_3]=\frac{5}{6}C^2 \sum_{m,n,p} R^2 f''(CR)+\frac{8}{3}C\sum_{m,n,p}Rf'(CR)+\frac{14}{3}C^2\sum_{m,n,p}T f''(CR),\\
&\partial_{xx}^2 E_f[D_3]=\frac{C^2}{3}\sum_{m,n,p}R^2 f''(CR)+C\sum_{m,n,p}R f'(CR) +\frac{4}{3}C^2\sum_{m,n,p} Tf''(CR),\\
&\partial_{yy}^2 E_f[D_3]=\partial_{zz}^2 E_f[D_3]=\frac{2}{3}C^2\sum_{m,n,p}R^2f''(CR) +\frac{4}{3}C\sum_{m,n,p} Rf'(CR)-\frac{8}{3}C^2\sum_{m,n,p} T f''(CR),\\
&\partial_{uv}^2 E_f[D_3]=-\frac{C^2}{2}\sum_{m,n,p}R^2 f''(CR)-2C\sum_{m,n,p}R f'(CR)-6C^2\sum_{m,n,p}T f''(CR),\\
&\partial_{xy}^2 E_f[D_3]=-\frac{C^2}{3}\sum_{m,n,p}R^2 f''(CR)-\frac{2C}{3}\sum_{m,n,p} R f'(CR)+\frac{4C^2}{3}\sum_{m,n,p} Tf''(CR),\\
&\partial_{ux}^2 E_f[D_3]=\partial_{uy}^2 E_f[D_3]=\partial_{uz}^2 E_f[D_3]=\partial_{vx}^2 E_f[D_3]=\partial_{vy}^2 E_f[D_3]=\partial_{vz}^2 E_f[D_3]=\partial_{xz}^2 E_f[D_3]\\
&\quad\quad\quad\quad \hspace{1mm} =\partial_{yz}^2 E_f[D_3]=0.
\end{align*}
\end{prop}

\section{Theta function and completely monotone potentials}\label{sectheta}

Ennola \cite{Ennola} proved that $D_3$ is a local minimizer for the Epstein zeta function, at any fixed density and for any $s>0$. A first consequence, by duality \eqref{fcteqzeta}, is the same result for $D_3^*$.\\

In this part, we consider the potential defined by $f_\alpha(r)=e^{-\alpha r}$, $\alpha>0$. The associated energy is the so-called theta function
$$
\theta_L(\alpha):=E_{f_\alpha}[L]+1=\sum_{p\in L}e^{-\alpha |p|^2}.
$$
As we recalled in \cite{BetTheta15}, there is a strong relation between $\theta_L$ and $E_f[L]$ if $f$ is completely monotone, i.e. for any $r>0$ and any $k\in\N$, $(-1)^k f^{(k)}(r)>0$. By Hausdorff-Bernstein-Widder theorem \cite{Bernstein}, $f$ is completely monotone if and only if $f$ is the Laplace transform of a finite positive measure $\mu$, and therefore
$$
E_f[L]=\int_0^{+\infty}\left(\theta_L(t)-1\right)d\mu(t).
$$
Consequently, any optimality for any $\alpha>0$, of a lattice $L_0$, for $L\mapsto \theta_L(\alpha)$ give the same optimality for $E_f$ if $f$ is a completely monotone function $f$. In dimension $d\in\{2,4,8,24\}$, Cohn and Kumar \cite{CohnKumar} conjectured that the same lattice is the unique minimizer of $L\mapsto \theta_L(\alpha)$, at any fixed density. In dimension $d=3$, it turns out that this ``universality" is not true, as explained by Sarnak and Str\"ombergsson in \cite{SarStromb}. The goal of the following study is to prove the instability of $\Z^3$ and the local (non-)minimality of $D_3$ and $D_3^*$ with respect to the value of the parameter $\alpha>0$.

\subsection{Non-optimality of simple cubic lattice}
As already understood by Born \cite{born1940}, the Simple Cubic lattices are all unstable, in particular for every completely monotone potential (Gaussian interaction, inverse power laws, etc.).
\begin{prop}\label{Z3saddleptLJ}
Let $f_\alpha(r)=e^{-\alpha r}$. For any $\alpha>0$ and any $V>0$, $\Z^3$ is a saddle point of $L\mapsto E_{f_\alpha}[L]$ among Bravais lattices of fixed volume $V$. Thus, $\Z^3$ is a saddle point of $L\mapsto E_{f}[L]$ among Bravais lattices of fixed volume $V$ if $f\in \mathcal{F}$ is completely monotone.
\end{prop}
\begin{proof} We write $f=f_\alpha$ and we use Proposition \ref{2ndZ3}. We first prove the following result: for any $\alpha>0$,
$$
\partial_{xx}^2 E_f[\Z^3]=\partial_{yy}^2 E_f[\Z^3]=\partial_{zz}^2 E_f[\Z^3]<0.
$$
Letting $\beta=C\alpha 2^{-1/3}$, we have
\begin{align*}
\partial_{xx}^2 E_f[\Z^3] &=2^{4/3}C^2\alpha^2\sum_{m,n,p}n^2p^2e^{-\alpha C I}-2^{2/3}C\alpha \sum_{m,n,p}p^2 e^{-\alpha C I}\\
&=2^{2/3}C\alpha\sum_{m,n,p}\left( 2^{2/3}C \alpha n^2p^2-p^2 \right)e^{-\alpha C I}\\
& =2\beta\sum_{m,n,p}\left(2\beta n^2p^2-p^2  \right)e^{-\alpha C I}\\
&=2\beta \left(\sum_p p^2 e^{-\beta p^2}\right)^2 \left( \sum_n (2\beta n^2-1)e^{-\beta n^2} \right).
\end{align*}
We now apply \cite[Lemma (Fact 2)]{Faulhuber:2016aa} which states that, for any $s>0$,
\begin{equation}\label{FS1}
s\frac{\theta_3'(s)}{\theta_3(s)}+\frac{1}{s}\frac{\theta_3'(1/s)}{\theta_3(1/s)}=-1/2,
\end{equation}
where $\displaystyle \theta_3(s):=\sum_{k\in \Z} e^{-\pi k^2 s}$ is the classical one-dimensional theta function. Then, applying \eqref{FS1} to $s=\beta/\pi$, we get, as $\theta_3'(1/s)<0$,
\begin{equation}\label{UseFS1}
2\beta \frac{\sum_n n^2 e^{-\beta n^2}}{\sum_n e^{-\beta n^2}}<1,
\end{equation}
which is equivalent with $\sum_n (2\beta n^2-1)e^{-\beta n^2}<0$, and that proves $\partial_{xx}^2 E_f[\Z^3]<0$.\\
Let us now prove that $\partial_{uu}^2 E_f[\Z^3]>0$. Letting $\beta=C\alpha 2^{-1/3}$, we have
\begin{align*}
&\partial_{uu}^2 E_f[\Z^3]=\frac{3}{2^{1/3}}C^2\alpha^2\sum_{m,n,p} (p^4-p^2n^2)e^{-\alpha CI}-4C\alpha\sum_{m,n,p}p^2 e^{-\alpha CI}\\
&=2^{1/3}\beta\left( \sum_m e^{-\beta m^2} \right)\sum_{n,p}\left(3\beta(p^4-p^2n^2)-p^2  \right)e^{-\beta n^2}e^{-\beta p^2}.
\end{align*}
Now, we expand the third factor and we write it in terms of $\theta_3$:
\begin{align*}
&\sum_{n,p}\left(3\beta(p^4-p^2n^2)-p^2  \right)e^{-\beta n^2}e^{-\beta p^2}=\frac{3\beta}{\pi^2}\theta_3''(\beta/\pi)\theta(\beta/\pi)-\frac{3\beta}{\pi^2}\left( \theta_3'(\beta/\pi) \right)^2+\frac{1}{\pi}\theta_3(\beta/\pi)\theta_3'(\beta/\pi).
\end{align*}
The ``refined logarithmic convexity" \cite[Th. 2.3]{Faulhuber:2016aa} states that, for any $s>0$,
$$
\theta_3''(s)\theta_3(s)-\theta_3'(s)^2>-\frac{\theta_3'(s)\theta_3(s)}{s}>0.
$$
We apply it for $s=\beta/\pi$ and we find
$$
\sum_{n,p}\left(3\beta(p^4-p^2n^2)-p^2  \right)e^{-\beta n^2}e^{-\beta p^2}>-\frac{2}{\pi}\theta_3'(\beta/\pi)\theta_3(\beta/\pi)>0.
$$
It follows that $\partial_{uu}^2 E_f[\Z^3]>0$. The same arguments work to prove $\partial_{vv}^2 E_f[\Z^3]>0$. 
\end{proof}

\begin{remark}
Even though we are studying a toy static model for interaction on condensed matter, the fact that there are few elements (actually, only the Polonium) that can have a Simple Cubic structure in its solid state is not surprising. 
\end{remark}

\subsection{Optimality of FCC and BCC with respect to $\alpha$}

In \cite[Thm. 1.7]{BeterminPetrache}, we proved that $D_3$ (resp. $D_3^*$) is not local a minimizer of $L\mapsto \theta_L(\alpha)$  if $\alpha$ is enough small (resp. enough large). Furthermore, we showed that these lattices are local minimizers for some values of $\alpha$ (we computed it for only a finite number of values of $\alpha$). For that, we used both Montgomery \cite{Mont} and Baernstein \cite{baernstein} results combined with our minimization result among body-centred-orthorhombic lattices. In this part, we show that the decomposition we proved in Proposition \ref{2ndD3} allows to prove the local minimality of $D_3$ (resp. $D_3^*$ ) for enough large (resp. enough small) values of $\alpha$.\\

We begin by a lemma which is the analogue of \cite[Lemma 5]{Ennola} for the theta function.

\begin{lemma}
For any $\beta>0$,
\begin{equation}\label{sumTR}
\sum_{m,n,p} T e^{-\beta R}>0.
\end{equation}
\end{lemma}
\begin{proof}
We use \cite[Lemma 5]{Ennola} which states that, for any $t\in \N\backslash \{0,1\}$,
$$
A(t):=\sum_{m,n,p\atop R(m,n,p)\leq t} T(m,n,p)>0.
$$
Then, as in \cite[Lemma 10]{Ennola} we get
$$
\sum_{m,n,p} Te^{-\beta R}=\sum_{t=2}^{+\infty} (A(t)-A(t-1))e^{-\beta t}=\sum_{t=2}^{+\infty} A(t)\left( e^{-\beta t}-e^{-\beta(t+1)} \right)>0
$$
because $A(1)=0$.
\end{proof}

Now, using this lemma and our decomposition in Proposition \ref{2ndD3} in terms of $R$ and $T$ of the second derivatives, we prove the local optimality of $D_3$ for extremal values of $\alpha$. The local optimality of $D_3^*$, obtained by duality, will be stated in a corollary below. We recall that, by the scaling property of the theta function, it is sufficient to state the following result for Bravais lattices of unit volume.

\begin{prop}\label{locmintheta}
There exist $\alpha_0$ and $\alpha_1$ such that:
\begin{enumerate}
\item for any $0<\alpha<\alpha_0$, $D_3$ is a saddle point of $E_{f_\alpha}$ among Bravais lattices of fixed volume $1$;
\item for any $\alpha>\alpha_1$, $D_3$ is a local minimizer of $E_{f_\alpha}$ among Bravais lattices of fixed volume $1$.
\end{enumerate}
\end{prop}
\begin{proof}
We use Proposition \ref{2ndD3}. Indeed, we easily compute, with $\beta= C\alpha$,
\begin{align}
&\label{theta1}\partial_{uu}^2 E_{f_\alpha}[D_3]=\frac{\beta}{2}\sum_{m,n,p}\left[\beta(R^2+12T)-4R   \right]e^{-\beta R},\\
&\label{theta2}\partial_{xx}^2 E_{f_\alpha}[D_3]=\frac{\beta}{3}\sum_{m,n,p}\left[ \beta(R^2+4T)-3R \right]e^{-\beta R},\\
&\label{theta3}\partial_{zz}^2 E_{f_\alpha}[D_3]=\frac{2\beta}{3}\sum_{m,n,p}\left[\beta(R^2-4T)-2R  \right]e^{-\beta R}.
\end{align}
By \eqref{sumTR}, we have $\sum (R^2+12T)e^{-\beta R}>0$ and $\sum (R^2+4T)e^{-\beta R}>0$. Therefore, it is clear that \eqref{theta1} and \eqref{theta2} are negative for $\beta$ small enough and positive for $\beta$ large enough because their nearest-neighbours terms dominate the rest. Since we have \eqref{ineqRT}, it is the same for \eqref{theta3}.\\
Furthermore, defining the following positive quantities (since $R> 0$ and \eqref{sumTR}),
$$
A_1=\sum_{m,n,p} R^2 e^{-\beta R},\quad A_2=\sum_{m,n,p} R e^{-\beta R}, \quad\textnormal{and}\quad A_3=\sum_{m,n,p}T e^{-\beta R},
$$
we have
\begin{align*}
&\partial_{uu} E_{f_\alpha}[D_3]\partial_{vv}  E_{f_\alpha}[D_3]-\left(\partial_{uv} E_{f_\alpha}[D_3]\right)^2\\
&=\frac{1}{6}(A_1+12A_3)(A_1-4A_3)\beta^4-\frac{1}{3}\left(3A_1A_2+4A_2A_3  \right)\beta^3+\frac{4}{3}A_2^2 \beta^2,
\end{align*}
and
\begin{align*}
&\partial_{xx} E_{f_\alpha}[D_3]\partial_{yy}  E_{f_\alpha}[D_3]-\left(\partial_{xy} E_{f_\alpha}[D_3]\right)^2\\
&=\frac{1}{9}(A_1+12A_3)(A_1-4A_3)\beta^4-\frac{2}{9}\left(3A_1A_2+4A_2A_3  \right)\beta^3+\frac{8}{9}A_2^2 \beta^2.
\end{align*}
Therefore, as $4T<R^2$, we get $A_1>4A_3$ and, because all the $A_i$'s are of the same order with respect to $\beta$, if $\beta$ is enough large, then both above quantities are positive. Thus, we get the local minimality of $D_3$ for $\beta$ enough large. Moreover, if $\beta$ is enough small, these quantities are positive and $D_3$ is a saddle point.
\end{proof}

By duality, we get the same kind of result for the BCC lattice.
\begin{corollary}\label{corlocmintheta}
Let $\alpha_0$ and $\alpha_1$ be as in the previous proposition, then:
\begin{enumerate}
\item for any $\alpha>1/\alpha_0$, $D_3^*$ is a saddle point of $E_{f_\alpha}$ among Bravais lattices of fixed volume $1$;
\item for any $0<\alpha<1/\alpha_1$, $D_3^*$ is a local minimizer of $E_{f_\alpha}$ among Bravais lattices of fixed volume $1$.
\end{enumerate}
\end{corollary}

\section{Application to Lennard-Jones-type potentials}\label{secLJ}

In this part we consider the Lennard-Jones-type potentials defined on $(0,+\infty)$ by
\begin{equation}\label{LJpot}
f(r)=\frac{a_2}{r^{x_2}}-\frac{a_1}{r^{x_1}},\quad (a_1,a_2)\in(0,+\infty)^2,\quad \quad 3/2<x_1<x_2.
\end{equation}
This potential is known as a good model for interaction between dipoles (see \cite[Sec. 6.3]{BetTheta15} for examples). Furthermore, for $\zeta_{2^{-1/3}D_3}$ and $Y$ defined respectively by \eqref{zetaR} and \eqref{defY}, we let, as in \cite{Ennola},
\begin{align}
&\label{G} G(s):=s(s-3)\zeta_{2^{-1/3}D_3}(2s)+12s(s+1)Y(s),\\
&\label{H} H(s):=s(s-1)\zeta_{2^{-1/3}D_3}(2s)-4s(s+1)Y(s).
\end{align}
By \cite[Lemma 11 and Lemma 12]{Ennola}, we know that $G(s)>0$ and $H(s)>0$ for any $s>0$.\\

First of all, we prove the local minimality of $\Z^3$, for some values of the volume $V$, in the classical Lennard-Jones case $a_1=2$, $a_2=1$, $x_1=3$ and $x_2=6$. Second, we prove the local minimality of $D_3$ and $D_3^*$ at high density for any Lennard-Jones-type potential. Finally, we prove the local minimality of the FCC and BCC lattices among all the Bravais lattices, without volume restriction.

\subsection{Local optimality of the simple cubic lattice for the classical Lennard-Jones potential}

Here, we just consider the classical Lennard-Jones potential
$$
f(r^2)=\frac{1}{r^{12}}-\frac{2}{r^6}.
$$

\begin{prop}\label{Z3LJ}
Let $a_1=2$, $a_2=1$, $x_1=3$ and $x_2=6$. Then, there exists $V_1,V_2$ such that, for any $V_1<V<V_2$, $\Z^3$ is a local minimizer of $E_f$ among Bravais lattices of fixed volume $V$, where $V_1\approx 1.200$ and $V_2\approx 1.344$. Furthermore, if $V\not\in [V_1,V_2]$, then $\Z^3$ is a saddle point of $E_f$  among Bravais lattices of fixed volume $V$.
\end{prop}

\begin{proof}
We begin by computing
\begin{align*}
&\partial_{uu}^2E_f[\Z^3]=\frac{a_2 x_2}{C^{x_2}}2^{\frac{x_2+1}{3}}h_1(x_2)-\frac{a_1 x_1}{C^{x_1}}2^{\frac{x_1+1}{3}}h_1(x_1),\\
&\partial_{vv}^2E_f[\Z^3]=\frac{2^{3+x_2/3}a_2x_2}{C^{x_2}}h_2(x_2)-\frac{2^{3+x_1/3}a_1x_1}{C^{x_1}}h_2(x_1),\\
&\partial_{xx}^2E_f[\Z^3]=\partial_{yy}^2E_f[\Z^3]=\partial_{zz}^2E_f[\Z^3]=\frac{2^{1+x_2/3}a_2x_2}{C^{x_2}}h_3(x_2)-\frac{2^{1+x_1/3}a_1x_1}{C^{x_1}}h_3(x_1),\\
&\partial_{uv}^2E_f[\Z^3]=\frac{2^{\frac{x_2+2}{3}}a_2x_2}{C^{x_2}}h_4(x_2)-\frac{2^{\frac{x_1+2}{3}}a_1x_1}{C^{x_1}}h_4(x_1),
\end{align*}
where $C=V^{2/3} 2^{1/3}$ and
\begin{align*}
&h_1(x)=3(x+1)\sum_{m,n,p}\frac{p^4-p^2n^2}{(m^2+n^2+p^2)^{x+2}}-4\sum_{m,n,p}\frac{p^2}{(m^2+n^2+p^2)^{x+1}},\\
&h_2(x)=(x+1)\sum_{m,n,p}\frac{p^4-p^2n^2}{(m^2+n^2+p^2)^{x+2}}-\sum_{m,n,p}\frac{p^2}{(m^2+n^2+p^2)^{x+1}},\\
&h_3(x)=2(x+1)\sum_{m,n,p}\frac{p^2n^2}{(m^2+n^2+p^2)^{x+2}}-\sum_{m,n,p}\frac{p^2}{(m^2+n^2+p^2)^{x+1}},\\
&h_4(x)=-2\sum_{m,n,p}\frac{p^4-p^2n^2}{(m^2+n^2+p^2)^{x+2}}+3\sum_{m,n,p}\frac{p^2}{(m^2+n^2+p^2)^{x+1}}.
\end{align*}
Therefore, we easily show that there exists $V_1,V_2$ and $V_3$ such that
\begin{align*}
&\partial_{uu}^2E_f[\Z^3]>0 \iff V<V_3,\\
&\partial_{xx}^2E_f[\Z^3]=\partial_{yy}^2E_f[\Z^3]=\partial_{zz}^2E_f[\Z^3]>0 \iff V>V_1,\\
&\partial_{uu}^2E_f[\Z^3]\partial_{vv}^2E_f[\Z^3]-\left(\partial_{uv}^2E_f[\Z^3]\right)^2>0 \iff V<V_2\quad \textnormal{or}\quad V>V_4,
\end{align*}
where
$$
V_3\approx 1.482,\quad V_1\approx 1.200,\quad V_2\approx 1.344, \quad V_4\approx 1.5797.
$$
\end{proof}

\begin{remark}
We notice that $\Z^3$ is never a local maximum of the classical Lennard-Jones energy, as $\Z^2$ in two dimensions (see \cite{Beterloc}). Furthermore, the same result is true (with different values of $V_1$ and $V_2$) for any Lennard-Jones-type potential given by
$$
f_p(r)=\frac{1}{r^{2p}}-\frac{2}{r^p},\quad p>3/2,
$$
and the interval's length for the local minimality of $\Z^3$ goes to $0$ as $p\to+\infty$, i.e. $|V_2(p)-V_1(p)|\to 0$ as $p\to +\infty$.
\end{remark}

\subsection{Local optimality of the FCC and BCC lattices at high density}

We now study the local optimality of $D_3$ and $D_3^*$ for the Lennard-Jones-type potentials. As in \cite{Betermin:2014fy,BetTheta15,Beterloc} for the triangular lattice, the fact that these lattices are local minimizers for any inverse power law gives their local minimality at high density for the Lennard-Jones-type potentials.

\begin{prop}[Local optimality of FCC and BCC lattices]\label{locoptFCCLJ}
Let $f$ be a Lennard-Jones-type potential defined by \eqref{LJpot}. If we have
\begin{equation}\label{locminFCC}
V<\frac{1}{\sqrt{2}} \min\left\{ \left( \frac{a_2G(x_2)}{a_1G(x_1)} \right)^{\frac{3}{2(x_2-x_1)}}, \left( \frac{a_2H(x_2)}{a_1H(x_1)} \right)^{\frac{3}{2(x_2-x_1)}} \right\},
\end{equation}
then $D_3$ and $D_3^*$ local minimizers of $E_f$ among Bravais lattices of fixed volume $V$. Furthermore, if 
\begin{equation}\label{locmaxFCC}
V> \frac{1}{\sqrt{2}} \max\left\{ \left( \frac{a_2G(x_2)}{a_1G(x_1)} \right)^{\frac{3}{2(x_2-x_1)}}, \left( \frac{a_2H(x_2)}{a_1H(x_1)} \right)^{\frac{3}{2(x_2-x_1)}} \right\},
\end{equation}
then $D_3$ and $D_3^*$ are local maximizers  of $E_f$ among Bravais lattices of fixed volume $V$.
\end{prop}

\begin{proof}
We apply directly \cite[Lemma 4]{Ennola} and we get
$$
E_f[L]-E_f[D_3]=\frac{1}{12}\left( \frac{a_2G(x_2)}{C^{x_2}}-\frac{a_1G(x_1)}{C^{x_1}} \right)Q_1(u,v,x) +\frac{1}{12}\left( \frac{a_2H(x_2)}{C^{x_2}}-\frac{a_1H(x_1)}{C^{x_1}} \right)Q_2(v,x,y,z)+\Delta,
$$
where $Q_1(u,v,x)=3(u-v)^2+x^2$, $Q_2(v,x,y,z)=2(v-1)^2+(2z-1)^2+(x-2y+1)^2$ and 
$$
\lim_{\gamma\to 0}\frac{\Delta}{\gamma^2}=0
$$
for $\gamma=\max\left\{ |u-1|, |v-1|, |y-1/2|, |z-1/2| \right\}$.\\
Therefore, as $G(s)>0$ and $H(s)>0$ for any $s>0$, \eqref{locminFCC} is equivalent with $\frac{a_2G(x_2)}{C^{x_2}}-\frac{a_1G(x_1)}{C^{x_1}}>0$ and $\frac{a_2H(x_2)}{C^{x_2}}-\frac{a_1H(x_1)}{C^{x_1}}>0$, and then $D_3$ is a local minimizer of $E_f$ among Bravais lattices of fixed volume $V$ satisfying \eqref{locminFCC}. In the second case, \eqref{locmaxFCC} implies that the quantities are negative and $D_3$ is a local maximizer of $E_f$.\\
Using functional equation \eqref{fcteqzeta}, we easily find the same result for the local optimality of the $BCC$ lattice. 
\end{proof}
%
%Let us define
%$$
%h(x)=\frac{\pi^{2x}\Gamma(3/2-x)}{\pi^{3/2}\Gamma(x)},
%$$
%where $\Gamma(3/2-x)$ is given by the analytic extension of the Euler Gamma function.
%
%\begin{corollary}[Local optimality of BCC lattice]
%Let $f$ be a Lennard-Jones-type potential defined by \eqref{LJpot}. If we have
%\begin{equation}\label{locminFCC}
%V<\frac{1}{\sqrt{2}} \min\left\{ \left( \frac{a_2 h(x_2)G(3/2-x_2)}{a_1h(x_1)G(3/2-x_1)} \right)^{\frac{3}{2(x_2-x_1)}}, \left( \frac{a_2h(x_2)H(3/2-x_2)}{a_1h(x_1)H(3/2-x_1)} \right)^{\frac{3}{2(x_2-x_1)}} \right\},
%\end{equation}
%then $D_3^*$ is a local minimizer of $E_f$ among Bravais lattices of fixed volume $V$. Furthermore, if 
%\begin{equation}\label{locmaxFCC}
%V> \frac{1}{\sqrt{2}} \max\left\{ \left( \frac{a_2h(x_2)G(3/2-x_2)}{a_1h(x_1)G(3/2-x_1)} \right)^{\frac{3}{2(x_2-x_1)}}, \left( \frac{a_2h(x_2)H(3/2-x_2)}{a_1h(x_1)H(3/2-x_1)} \right)^{\frac{3}{2(x_2-x_1)}} \right\},
%\end{equation}
%then $D_3^*$ is a local maximizer  of $E_f$ among Bravais lattices of fixed volume $V$.
%\end{corollary}

\begin{example}\label{exampleclassic}
If $a_1=2$, $a_2=1$, $x_1=3$ and $x_2=6$, then, by Proposition \ref{locoptFCCLJ}, we get that:
\begin{enumerate}
\item if $0<V<V_{\min}$, $V_{\min}\approx 1.091$, then $D_3$ and $D_3^*$ are local minimizers of $E_f$ among Bravais lattices of fixed volume $V$;
\item if $V_{\min}<V<V_{\max}$, $V_{\max}\approx 1.313$, then $D_3$ and $D_3^*$ are saddle points of $E_f$ among Bravais lattices of fixed volume $V$;
\item if $V>V_{\max}$, then $D_3$ and $D_3^*$ are local maximizers of $E_f$ among Bravais lattices of fixed volume $V$.
\end{enumerate}
Furthermore, we compute the optimal volume $V^*(L)$ minimizing $V\mapsto E_f[V^{1/3}L]$ for $L\in \{\Z^3, D_3,D_3^*\}$. Let $\lambda=V^{1/3}$ and $g(\lambda):=E_f[\lambda L]=\zeta_L(12)\lambda^{-12}-2\zeta_L(6)\lambda^{-6}$. Then we have
\begin{equation}
g'(\lambda)\geq 0 \iff \lambda\geq \left( \frac{\zeta_L(12)}{\zeta_L(6)} \right)^{1/6} \iff V\geq \sqrt{\frac{\zeta_L(12)}{\zeta_L(6)}}=:V^*(L).
\end{equation}
\end{example}

We then get the following result:
\begin{prop}\label{prop-locglobalLJ}
Let $f$ be defined by \eqref{LJpot} with $a_1=2$, $a_2=1$, $x_1=3$ and $x_2=6$. Let $V^*(\Z^3)^{1/3}\Z^3$ (resp. $V^*(D_3)^{1/3}D_3$ and $V^*(D_3^*)^{1/3}D_3^*$) be the minimizer of $E_f$ among all the dilated of $\Z^3$ (resp. $D_3$ and $D_3^*$) with volume $V^*(\Z^3)$ (resp. $V^*(D_3)$ and $V^*(D_3^*)$). Then we have:
\begin{enumerate}
\item  $V^*(\Z^3)^{1/3}\Z^3$ is not a local minimizer of $E_f$ among all the Bravais lattices;
\item $V^*(D_3)^{1/3}D_3$ and $V^*(D_3^*)^{1/3}D_3^*$ are local minimizers of $E_f$ among all the Bravais lattices.
\end{enumerate}
\end{prop}
\begin{proof}
It is sufficient to compute $V^*(L)$ for $L\in \{\Z^3, D_3,D_3^*\}$:
\begin{itemize}
\item[-] $V^*(\Z^3)\approx 0.859 \not\in [V_1,V_2]$, where $V_1$ and $V_2$ are given by Proposition \ref{Z3LJ}. Therefore, by Proposition \ref{Z3LJ}, $V^*(\Z^3)^{1/3}\Z^3$ is not a local minimizer of $E_f$ among all the Bravais lattices.
\item[-] $V^*(D_3)\approx 0.648<V_{\min}\approx 1.091$, therefore by Proposition \ref{locoptFCCLJ}, $V^*(D_3)^{1/3}D_3$ is a local minimizer of $E_f$ among all the Bravais lattices;
\item[-] $V^*(D_3^*)\approx 0.664<V_{\min}$, therefore by Proposition \ref{locoptFCCLJ}, $V^*(D_3^*)^{1/3}D_3^*$ is a local minimizer of $E_f$ among all the Bravais lattices.
\end{itemize}
\end{proof}

%Actually, thanks to the work of Misra \cite{misra1940}, we know, for $x_2>5/2$, that:
%\begin{enumerate}
%\item the minimizer of $E_f$ among all the dilated of $D_3$ is a local minimizer among all the Bravais lattices (without a density constraint). In particular, this result implies that this is the same for the triangular lattice in dimension $d=2$;
%\item the minimizer of $E_f$ among all the dilated of $D_3^*$ is not a local minimizer among all the Bravais lattices;
%\item the minimizer of $E_f$ among all the dilated of $\Z^3$ is not a local minimizer among all the Bravais lattices.
%\end{enumerate}

Thus, from this result and Misra's work \citep{misra1940}, the good candidate for the global minimization of $E_f$ with a Lennard-Jones-type interacting potential $f$ and $x_2>x_1>3/2$ is the FCC lattice $V^*(D_3)^{1/3}D_3$. If $0<x_1<x_2<3/2$, we expect, from Sarnak-Str\"ombergsson conjecture \cite{SarStromb}, that the BCC lattice $V^*(D_3^*)^{1/3}D_3^*$ is the good candidate.

\section{Appendix: General formulas for the second derivatives of $E_f$}\label{secAppendix}

We give the formulas for the second derivatives of $E_f$, for any $f$ and at any Bravais lattice $L=(u,v,x,y)$ with volume $V$.

\begin{lemma}[Second derivatives of the energy]\label{2nd} For any $f\in \mathcal{F}$, any fixed volume $V>0$ and any $L=(u,v,x,y)$, we have
\begin{align*}
&\partial^2_{uu}E_f[L]=C^2\sum_{m,n,p}\left(-\frac{(m+xn+yp)^2}{u^2}-\frac{v^2(n+zp)^2}{u^2}+\frac{u}{v^2}p^2  \right)^2 f''(Q_L)\\
&\quad\quad\quad\quad\quad\quad  +2C\sum_{m,n,p}\left(  \frac{(m+xn+yp)^2}{u^3}+\frac{v^2(n+zp)^2}{u^3}+\frac{1}{2v^2}p^2 \right) f'(Q_L),\\
&\partial^2_{vv}E_f[L]=C^2\sum_{m,n,p}\left(\frac{2v(n+zp)^2}{u}-\frac{u^2}{v^3}p^2 \right)^2 f''(Q_L) +C\sum_{m,n,p}\left(  \frac{2(n+zp)^2}{u}+\frac{3u^2}{v^4}p^2 \right) f'(Q_L),\\
&\partial_{xx}^2 E_f[L]=\frac{4C^2}{u^2}\sum_{m,n,p}n^2(m+xn+yp)^2f''(Q_L)+\frac{2C}{u}\sum_{m,n,p}n^2 f'(Q_L),\\
&\partial_{yy}^2 E_f[L]=\frac{4C^2}{u^2}\sum_{m,n,p}p^2(m+xn+yp)^2f''(Q_L)+\frac{2C}{u}\sum_{m,n,p}p^2 f'(Q_L),\\
&\partial_{zz}^2 E_f[L]=\frac{4C^2v^4}{u^2}\sum_{m,n,p}p^2(n+zp)^2f''(Q_L)+\frac{2Cv^2}{u}\sum_{m,n,p}p^2 f'(Q_L),\\
&\partial_{uv}^2 E_f[L]=C^2\sum_{m,n,p}\left(-\frac{(m+xn+yp)^2}{u^2}-\frac{v^2(n+zp)^2}{u^2}+\frac{u}{v^2}p^2   \right)\left( \frac{2v(n+zp)^2}{u}-\frac{u^2}{v^3}p^2 \right)f''(Q_L)\\
&\quad\quad\quad\quad\quad\quad  -2C\sum_{m,n,p}\left( \frac{v(n+zp)^2}{u^2}+\frac{u}{v^3}p^2 \right)f'(Q_L),\\
&\partial_{ux}^2 E_f[L]=\frac{2C^2}{u}\sum_{m,n,p}n(m+xn+yp)\left(-\frac{(m+xn+yp)^2}{u^2}-\frac{v^2(n+zp)^2}{u^2}+\frac{u}{v^2}p^2    \right)f''(Q_L)\\
&\quad\quad\quad\quad\quad\quad -\frac{2C}{u^2}\sum_{m,n,p}n(m+xn+yp)f'(Q_L),\\
&\partial_{uy}^2 E_f[L]=\frac{2C^2}{u}\sum_{m,n,p}p(m+xn+yp)\left(-\frac{(m+xn+yp)^2}{u^2}-\frac{v^2(n+zp)^2}{u^2}+\frac{u}{v^2}p^2    \right)f''(Q_L)\\
&\quad\quad\quad\quad\quad\quad -\frac{2C}{u^2}\sum_{m,n,p}p(m+xn+yp)f'(Q_L),\\
&\partial_{uz}^2 E_f[L]=\frac{2C^2v^2}{u}\sum_{m,n,p}p(n+zp)\left(-\frac{(m+xn+yp)^2}{u^2}-\frac{v^2(n+zp)^2}{u^2}+\frac{u}{v^2}p^2    \right)f''(Q_L)\\
&\quad\quad\quad\quad\quad\quad -\frac{2Cv^2}{u^2}\sum_{m,n,p}p(n+zp)f'(Q_L),\\
&\partial_{vx}^2 E_f[L]=\frac{2C^2}{u}\sum_{m,n,p}n(m+xn+yp)\left( \frac{2v(n+zp)^2}{u}-\frac{u^2}{v^3}p^2 \right)f''(Q_L),\\
&\partial_{vy}^2 E_f[L]=\frac{2C^2}{u}\sum_{m,n,p}p(m+xn+yp)\left( \frac{2v(n+zp)^2}{u}-\frac{u^2}{v^3}p^2 \right)f''(Q_L),\\
&\partial_{vz}^2 E_f[L]=\frac{2C^2v^2}{u}\sum_{m,n,p}p(n+zp)\left( \frac{2v(n+zp)^2}{u}-\frac{u^2}{v^3}p^2 \right)f''(Q_L)+\frac{4Cv}{u}\sum_{m,n,p}p(n+zp)f'(Q_L),\\
&\partial_{xy}^2 E_f[L]=\frac{4C^2}{u^2}\sum_{m,n,p}np(m+xn+yp)^2f''(Q_L)+\frac{2C}{u}\sum_{m,n,p}npf'(Q_L),\\
&\partial_{xz}^2 E_f[L]=\frac{4C^2v^2}{u^2}\sum_{m,n,p}np(n+zp)(m+xn+yp)f''(Q_L),\\
&\partial_{yz}^2 E_f[L]=\frac{4C^2v^2}{u^2}\sum_{m,n,p}p^2(n+zp)(m+xn+yp)f''(Q_L).
\end{align*}
\end{lemma}

\paragraph{Acknowledgements.} I'm grateful for the support of MAThematics Center Heidelberg (MATCH). I also wish to express my thanks to the anonymous referee for her/his suggestions.

\bibliographystyle{plain}
\bibliography{locmin3d}

\begin{thebibliography}{10}

\bibitem{baernstein}
A.~Baernstein~II.
\newblock A minimum problem for heat kernels of flat tori.
\newblock {\em Contemporary Mathematics}, 201:227--243, 1997.

\bibitem{Bernstein}
S.~Bernstein.
\newblock {Sur les fonctions absolument monotones}.
\newblock {\em Acta Math.}, 52:1--66, 1929.

\bibitem{Beterloc}
L.~B{\'e}termin.
\newblock Local variational study of 2d lattice energies and application to
  lennard-jones type interactions.
\newblock \textit{Preprint. arXiv:1611.07820}, 2016.

\bibitem{BetTheta15}
L.~B{\'e}termin.
\newblock {Two-dimensional Theta Functions and Crystallization among Bravais
  Lattices}.
\newblock {\em SIAM J. Math. Anal.}, 48(5):3236--3269, 2016.

\bibitem{BeterminPetrache}
L.~B{\'e}termin and M.~Petrache.
\newblock Dimension reduction techniques for the minimization of theta
  functions on lattices.
\newblock {\em J. Math. Phys.}, 58:071902, 2017.

\bibitem{Betermin:2014fy}
L.~B{\'e}termin and P.~Zhang.
\newblock {Minimization of energy per particle among Bravais lattices in
  $\R^2$: Lennard-Jones and Thomas-Fermi cases}.
\newblock {\em {Commun. Contemp. Math.}}, 17(6):1450049, 2015.

\bibitem{Blanc:2015yu}
X.~Blanc and M.~Lewin.
\newblock {The Crystallization Conjecture: A Review}.
\newblock {\em EMS Surveys in Mathematical Sciences}, 2:255--306, 2015.

\bibitem{born1940}
M.~Born.
\newblock On the stability of crystal lattices. i.
\newblock {\em Mathematical Proceedings of the Cambridge Philosophical
  Society}, 36(2):160--172, 1940.

\bibitem{Cassels}
J.W.S. Cassels.
\newblock {On a Problem of Rankin about the Epstein Zeta-Function}.
\newblock {\em Proceedings of the Glasgow Mathematical Association}, 4:73--80,
  7 1959.

\bibitem{CazorlaBoronatQC}
C.~Cazorla and J.~Boronat.
\newblock Simulation and understanding of atomic and molecular quantum
  crystals.
\newblock {\em Reviews of Modern Physics}, 89:035003, 2017.

\bibitem{CohnKumar}
H.~Cohn and A.~Kumar.
\newblock {Universally optimal distribution of points on spheres}.
\newblock {\em J. Amer. Math. Soc.}, 20(1):99--148, 2007.

\bibitem{Diananda}
P.~H. Diananda.
\newblock {Notes on Two Lemmas concerning the Epstein Zeta-Function}.
\newblock {\em Proceedings of the Glasgow Mathematical Association},
  6:202--204, 7 1964.

\bibitem{Eno2}
V.~Ennola.
\newblock {A Lemma about the Epstein Zeta-Function}.
\newblock {\em Proceedings of The Glasgow Mathematical Association},
  6:198--201, 1964.

\bibitem{Ennola}
V.~Ennola.
\newblock {On a Problem about the Epstein Zeta-Function}.
\newblock {\em Mathematical Proceedings of The Cambridge Philosophical
  Society}, 60:855--875, 1964.

\bibitem{Faulhuber:2016aa}
M.~Faulhuber and S.~Steinerberger.
\newblock Optimal gabor frame bounds for separable lattices and estimates for
  jacobi theta functions.
\newblock Preprint. Arxiv:1601.02972, 01 2016.

\bibitem{TheilFlatley}
L.~Flatley and F.~Theil.
\newblock {Face-Centred Cubic Crystallization of Atomistic Configurations}.
\newblock {\em Archive for Rational Mechanics and Analysis}, 219(1):363--416,
  2015.

\bibitem{Rad2}
R.~C. Heitmann and C.~Radin.
\newblock {The Ground State for Sticky Disks}.
\newblock {\em Journal of Statistical Physics}, 22:281--287, 1980.

\bibitem{misra1940}
R.~D. Misra.
\newblock On the stability of crystal lattices. ii.
\newblock {\em Mathematical Proceedings of the Cambridge Philosophical
  Society}, 36(2):173--182, 004 1940.

\bibitem{Mont}
H.~L. Montgomery.
\newblock {Minimal Theta Functions}.
\newblock {\em Glasgow Mathematical Journal}, 30(1):75--85, 1988.

\bibitem{Rad3}
C.~Radin.
\newblock {The Ground State for Soft Disks}.
\newblock {\em Journal of Statistical Physics}, 26(2):365--373, 1981.

\bibitem{Rankin}
R.~A. Rankin.
\newblock {A Minimum Problem for the Epstein Zeta-Function}.
\newblock {\em Proceedings of The Glasgow Mathematical Association},
  1:149--158, 1953.

\bibitem{SarStromb}
P.~Sarnak and A.~Str{\"o}mbergsson.
\newblock {Minima of Epstein's Zeta Function and Heights of Flat Tori}.
\newblock {\em Inventiones Mathematicae}, 165:115--151, 2006.

\bibitem{Suto1}
A.~S{\"u}to.
\newblock {Crystalline Ground States for Classical Particles}.
\newblock {\em Physical Review Letters}, 95, 2005.

\bibitem{Suto2}
A.~S{\"u}to.
\newblock {Ground State at High Density}.
\newblock {\em Communications in Mathematical Physics}, 305:657--710, 2011.

\bibitem{Crystal}
F.~Theil.
\newblock {A Proof of Crystallization in Two Dimensions}.
\newblock {\em Communications in Mathematical Physics}, 262(1):209--236, 2006.

\end{thebibliography}
\end{document}